\documentclass{article}
\usepackage{graphicx} 
\usepackage{qcircuit}
\usepackage[utf8]{inputenc}
\usepackage[english]{babel}
\usepackage[T1]{fontenc}
\usepackage{amsmath}
\usepackage[breaklinks=true]{hyperref}
\usepackage{ulem}

\usepackage{amssymb}
\usepackage{arydshln}

\usepackage{amsthm,scalerel}

\usepackage[matrix,frame,arrow]{xypic}

\newtheorem{theorem}{Theorem}[section]
\newtheorem{lemma}[theorem]{Lemma}
\newtheorem{corollary}{Corollary}[theorem]
\newtheorem{definition}{Definition}[theorem]

\title{Controlled Gates in the Clifford Hierarchy}
\author{Jonas T. Anderson \\ \href{mailto:sjonas.tyler.anderson@gmail.com}{jonas.tyler.anderson@gmail.com} \\
   \and Matthew Weippert}

\date{\today}

\begin{document}

\maketitle
\begin{abstract}
    In this note we prove a necessary set of conditions which must be satisfied by any controlled gate in the qubit Clifford Hierarchy. These conditions are straightforward to derive yet quite restricting. We also extend our proofs to gates composed of certain direct sums of unitaries. Finally, we provide some evidence that these conditions are also sufficient.
\end{abstract}

\section{Introduction}
The $n$-qubit Clifford Hierarchy \cite{Gottesman1999, Zeng2008} is recursively defined as 

\begin{equation}\label{CHdef}
    \mathcal{CH}_{k} \equiv \{U | U P U^\dagger \subseteq \mathcal{CH}_{k-1}, \forall P\in\mathcal{P}_n\}
\end{equation}
with the first level ($k=1$) defined as $\mathcal{CH}_1 \equiv \mathcal{P}_n$, the $n$-qubit Pauli group. $\mathcal{CH}_2$ is the $n$-qubit Clifford group. For $k\ge3$, the elements of $\mathcal{CH}_k$ no longer form a group. 

In \cite{Hu2021ClimbingHierarchy} it was noted that one can `climb' the Clifford Hierarchy by starting with non-trivial Pauli gates (rotations about the $X,Y,Z$ axis by $\pi$) and adding a control or taking a square root of the unitary rotation ({\it i.e.}~halving the rotation angle). Either of these operations produces a new gate that is one level higher in the Clifford Hierarchy. Here we look more generally at what conditions are necessary for a controlled gate to be in the Clifford Hierarchy. While the full structure of the Clifford Hierarchy is still not known, the necessary conditions we derive here apply generally to any controlled gate in the qubit Clifford Hierarchy. 

In section \ref{proofs} we provide proofs of two necessary conditions that a controlled-$U$ gate must satisfy to be in the Clifford Hierarchy. In section \ref{sufficient} we provide evidence that these conditions are also sufficient.

It is our intent that these insights can be combined with existing results to further elucidate the structure of the Clifford Hierarchy \cite{Hu2021ClimbingHierarchy, Anderson2024OnHierarchy, Anderson2016ClassificationCodes, Pllaha2020Un-Weyl-ingHierarchy, Rengaswamy2019UnifyingRings, Zeng2008,  Cui2017DiagonalHierarchy, Beigi2010C3Operations, Bengtsson2014OrderHierarchy,  deSilva2021EfficientDimensions, Liu2023ApproximateCorrection}.

\subsection{Fixing Notation}

In what follows, we denote a controlled-$U$ gate by ${\bf C}(U)$. We will write ${\bf C}(U) = |0\rangle\langle 0| \otimes I + |1\rangle\langle 1| \otimes U$ where the control is on the left side of the tensor product. We fix this for convenience, but our proofs still hold for control on other qubits, as SWAP gates (on physical qubits) are Clifford gates and (left or right) multiplication by Clifford gates does not change membership in the Clifford Hierarchy. Later when we depict ${\bf C}(U)$ in a circuit diagram we will always show the control on the topmost wire. Furthermore, we have fixed the control to be `on' when the control qubit is in the $|1\rangle$ state. Our proof holds if the control is on in the $|0\rangle$ state or for control in the $|+\rangle, |-\rangle,|+ i\rangle, |-i\rangle$ states as these on-off or basis changes are equivalent up to multiplication by Clifford gates.

\section{Proof of Necessary Conditions}\label{proofs}

In the following we will prove necessary conditions for membership in the Clifford Hierarchy for any block-diagonal unitary matrix consisting of two $2^{N-1}\times 2^{N-1}$ blocks. We will show that controlled-unitary gates are a subset of these matrices and give that proof as a corollary.  

Let $U = |0\rangle\langle 0| \otimes U_1 + |1\rangle\langle 1| \otimes U_2$ be a $2^N \times 2^N$ block-diagonal unitary matrix consisting of two $2^{N-1}\times 2^{N-1}$ blocks. The blocks $U_1$ and $U_2$ must be unitary since $U$ is unitary if and only if $U_1$ and $U_2$ are unitary.  

\begin{lemma}\label{lem1}$U \in \mathcal{CH}_1 \iff U_1 = \pm U_2$ and $U_1,U_2 \in \mathcal{CH}_1$. 
\end{lemma}

$\mathcal{CH}_1$ is the first level of the qubit Clifford Hierarchy. The elements form a group and consist of Pauli strings {\it i.e.}~tensor products of single-qubit Pauli matrices times a $U(1)$ phase (usually restricted to $i^k$ or $\pm 1$). The details of this phase are not important here and any of these definitions of $\mathcal{CH}_1$ will suffice for our proofs below.

\begin{proof}

$(\implies)$ 

Case 1: $U_1 = U_2 = P \in \mathcal{CH}_1$. Then $U = I\otimes P$ which is in $\mathcal{CH}_1$

Case 2: $U_1 = -U_2 = P \in \mathcal{CH}_1$. Then $U = Z\otimes P$ which is in $\mathcal{CH}_1$

$(\impliedby)$\footnote{This direction was proven by the Quantum StackExchange community \cite{Quantumcomputing.stackexchange.com:Matrices}. We are grateful for this community and hope it continues to thrive.} 

Assume $U$ equals a Pauli matrix and has the form $|0\rangle \langle 0| \otimes U_1 + |1 \rangle \langle 1| \otimes U_2$. The set of Pauli matrices $\mathcal{P}_N := \{I, X, Y, Z\}^{\otimes N}$ is an orthogonal basis for $2^N \times 2^N$ complex matrices, so there is exactly one $P \in \mathcal{P}_N$ for which $\text{Tr}(PU) \neq 0$. And up to a $U(1)$ phase, there is exactly one element of $P \in \mathcal{CH}_1$ with $\text{Tr}(PU) \neq 0$. 

For any $P$ of the form $P = X \otimes P_{N-1}$ (with $P_{N-1} \in \mathcal{P}_{N-1}$) you can compute
\begin{align*}
\text{Tr}(PU) &= \text{Tr}\left( (X \otimes P_{N-1})(|0\rangle \langle 0| \otimes U_1 + |1 \rangle \langle 1| \otimes U_2) \right)
\\&= \text{Tr}(X |0\rangle \langle 0|) \text{Tr}(P_{N-1} U_1) + \text{Tr}(X|1 \rangle \langle 1|)\text{Tr}( P_{N-1} U_2)
\\&=0,
\end{align*}
so $U$ cannot start with $X$. A similar proof follows for $Y$. There will be exactly one $P_{N-1}$ such that $U \in \{I_2 \otimes P_{N-1}, Z \otimes P_{N-1}\}$, {\it i.e.} $U_1 = P_{N-1}$ and $U_2 = \pm P_{N-1}$. 
\end{proof}

In what follows we use $\oplus$ to denote the direct sum of two matrices of equal size. We are currently working only with unitary matrices acting on some number of qubits; therefore, in the equation: $U=U_1\oplus U_2$, we have that $U_1,U_2\in SU(2^k)$ and $U\in SU(2^{k+1})$.

\begin{theorem}\label{thm1.1} $U_1\oplus U_2$ is in the Clifford Hierarchy only if $U_1$ and $U_2$ are in the Clifford Hierarchy. 
\end{theorem}

First, we will introduce some notation. Then, using Lemma \ref{lem1} above, we will prove the theorem.   
\begin{proof}
Let $U_1 \notin \mathcal{CH} \mbox{ or } U_2 \notin \mathcal{CH}$ and assume $U_1 \oplus U_2 \in \mathcal{CH}$ at some finite level. 

Define the following recursive function on tuples of Pauli strings, $P_i$, as follows:
\begin{center}
\begin{align}
    \mathcal{F}_U[P_1] &:= UP_1 U^{\dagger} \\
    \mathcal{F}_U[P_1,P_2] &:= UP_1 U^{\dagger}P_2UP_1 U^{\dagger} = \mathcal{F}_{UP_1U^{\dagger}}[P_2]. \\
    \mathcal{F}_U[P_1,...,P_k] &:= \mathcal{F}_{\mathcal{F}_U[P_1,...,P_{k-1}]}[P_k].
\end{align}
\end{center}

The following is equivalent to the standard definition of the Clifford Hierarchy: 
\begin{definition}
    If a gate $U$ is a member of $\mathcal{CH}$, then for any length tuple of Pauli strings, $\vec{P}$, $\mathcal{F}_{U}[\vec{P}]$ must be in $\mathcal{CH}$. 
    \end{definition}
    
Additionally, if $U \in \mathcal{CH}_k$ (the $k^{\text{th}}$ level of the Clifford Hierarchy), then after recursively applying $\mathcal{F}_U$ $k-1$ times all $\mathcal{F}_{\mathcal{F}_U[P_1,...,P_{k-1}]}[P_k,P_{k+1},...] \in \mathcal{CH}_1$. Conversely, If a gate $U$ is not a member of $\mathcal{CH}$ (at any finite level), then there must exist for any finite length a tuple of Pauli strings for which $\mathcal{F}_{U}[\vec{P}] \notin \mathcal{CH}$.

Recall that $U_1 \mbox{ or } U_2 \notin \mathcal{CH}$. Let $\vec{P}$ be a tuple of Pauli stings such that $\mathcal{F}_{U_i}[\vec{P}] \notin \mathcal{CH}$ for any length. Here $U_i$ denotes the unitary ($U_1$ or $U_2$) which is not in $\mathcal{CH}$. If both $U_1$ and $U_2$ are not in $\mathcal{CH}$ it can denote either unitary. Then define the tuple $I\otimes \vec{P} := (I\otimes P_1, I\otimes P_2, ...)$.

Now, if $U_1\oplus U_2$ is in $\mathcal{CH}$ at some finite level, we must have that $\mathcal{F}_{U_1\oplus U_2}[\vec{P}]$ is in $\mathcal{CH}$ for any tuple of Pauli strings. Furthermore, since $U_1\oplus U_2$ is in some finite level of the Clifford Hierarchy, we must eventually have that $\mathcal{F}_{\mathcal{F}_{U_1\oplus U_2}}[\vec{P}] \in \mathcal{CH}_1$ for all tuples of Pauli strings, $\vec{P}$. Choose $I\otimes \vec{P_k}$ to be of length long enough such that $\mathcal{F}_{U_1\oplus U_2}[I\otimes \vec{P_k}] \in \mathcal{CH}_1$. Next, we will examine how $\mathcal{F}_{U_1\oplus U_2}$ acts on Pauli strings of this form:

\begin{equation}
    \mathcal{F}_{U_1\oplus U_2}[I\otimes P] = (U_1\oplus U_2) I \otimes P (U_1^\dagger\oplus U_2^\dagger)
\end{equation}
which acts as follows on the blocks:
\begin{equation}
\begin{pmatrix}
U_1 & 0 \\
0 & U_2
\end{pmatrix}
\begin{pmatrix}
P & 0 \\
0 & P
\end{pmatrix}
\begin{pmatrix}
U_1^\dagger & 0 \\
0 & U_2^{\dagger}
\end{pmatrix}
= 
\begin{pmatrix}
U_1PU_1^\dagger & 0 \\
0 & U_2PU_2^{\dagger}
\end{pmatrix}. \\
\end{equation}

We can then see that $\mathcal{F}_{U_1\oplus U_2}[I\otimes \vec{P_k}]$ acts as follows: 
\begin{align}
    \mathcal{F}_{U_1\oplus U_2}[I\otimes \vec{P_k}] = 
\begin{pmatrix}
\mathcal{F}_{U_1}[\vec{P_k}] & 0 \\
0 & \mathcal{F}_{U_2}[\vec{P_k}]
\end{pmatrix}. \\
\end{align}

Since we have assumed that $U_1\oplus U_2\in\mathcal{CH}$ and that $\vec{P_k}$ is of sufficient length such that $\mathcal{F}_{U_1\oplus U_2}[I\otimes \vec{P_k}] \in \mathcal{CH}_1$, we have from Lemma~\ref{lem1} that $\mathcal{F}_{U_1}[\vec{P_k}] = \pm \mathcal{F}_{U_2}[\vec{P_k}] \in \mathcal{CH}_1$. But since $U_i\notin \mathcal{CH}$, $\mathcal{F}_{U_i}[\vec{P}]$ should not be in $\mathcal{CH}$ for any finite length of $\vec{P}$. We have arrived at a contradiction and must conclude that $U_1\oplus U_2\notin\mathcal{CH}$.
\end{proof}

\begin{corollary}
    \label{cor1} ${\bf C}(U)$, the controlled-$U$ gate, is in the Clifford Hierarchy only if $U$ is in the Clifford Hierarchy. 
\end{corollary}

\begin{proof}
Left and right multiplication by Clifford gates does not change the level of a gate (or membership) in the Clifford Hierarchy. Multiplying by SWAP gates which are Clifford, a controlled-unitary gate, ${\bf C}(U)$, can always be written as $I\oplus U$. In the proof of Theorem \ref{thm1.1} above set $U_1 = I$. Then it follows that $U_1 \oplus U_2$ is in $\mathcal{CH}$ at some finite level only if $U_1$ and $U_2 = U$ are in $\mathcal{CH}$ at some finite level. Since $U_1 = I$ is in $\mathcal{CH}_1$ (the Pauli group) we conclude that ${\bf C}(U)$ is in $\mathcal{CH}$ at some finite level only if $U_2 = U$ is in $\mathcal{CH}$ at some finite level. 
\end{proof}

In the next proof we will make use of the following identity which holds for any gate acting on qubits of the form: $U_1\oplus U_2$. We include it here as a lemma. 

Recall that $U_1, U_2$ are defined to be equally-sized $2^{k-1}$ unitary matrices. The lemma below holds for all such $U_1, U_2$ with the understanding that the Pauli string $X\otimes I$ is shorthand for $X\otimes(\bigotimes_{k-1}I)$.

\begin{lemma}\label{AB}
\begin{align*}
(U_1\oplus U_2)(X\otimes I)(U_1^\dagger\oplus U_2^\dagger) = &
\begin{pmatrix}
U_1 & 0 \\
0 & U_2
\end{pmatrix} 
\begin{pmatrix}
0 & I \\
I & 0
\end{pmatrix} 
\begin{pmatrix}
U_1^\dagger & 0 \\
0 & U_2^\dagger
\end{pmatrix} \\
= 
\begin{pmatrix}
U_1U_2^\dagger & 0 \\
0 & U_2 U_1^\dagger
\end{pmatrix} 
\begin{pmatrix}
0 & I \\
I & 0
\end{pmatrix}
= &
\begin{pmatrix}
U_1U_2^\dagger & 0 \\
0 & U_2 U_1^\dagger
\end{pmatrix} 
(X\otimes I).
\end{align*}
\end{lemma}

\begin{theorem}\label{thm2.1}
$U_1 \oplus U_2$ is in the Clifford Hierarchy only if $A=(U_1 U_2^\dagger)^{2^{m}}=\pm(U_2 U_1^\dagger)^{2^{m}}=B$ and both $A$ and $B$ are Pauli (a member of $\mathcal{CH}_1$) for some integer $m>0$.\end{theorem}

An $n$-qubit gate, $V$, is in $\mathcal{CH}$ at level $k$ if $VPV^\dagger$ is in $\mathcal{CH}$ at level $k-1$ for all $n$-qubit Pauli strings $P$. 

For a gate, $U$, which acts on $n$ qubits to be contained in the Clifford Hierarchy at level $k$ ($\mathcal{CH}_k$), it must have $U P_i U^\dagger\in\mathcal{CH}_{k-1}$ for all $n$-qubit Pauli strings $P_i$. This requirement is both necessary and sufficient for membership in $\mathcal{CH}_k$.  

\begin{proof}
Now, since a gate $U=U_1\oplus U_2$ must satisfy the same requirements for membership in $\mathcal{CH}_k$, we can see (using Lemma~\ref{AB} above) that for $U$ to be in $\mathcal{CH}_k$ (for some finite level $k$) it is necessary that 

\begin{equation*}
\begin{array}{cc}
\begin{pmatrix}
U_1U_2^\dagger & 0 \\
0 & U_2 U_1^\dagger
\end{pmatrix} 
(X\otimes I).
\end{array}
\in \mathcal{CH}_{k-1}.
\end{equation*}

But for this to be true it must also hold that 
\begin{equation*}
\begin{array}{cc}
\begin{pmatrix}
U_1U_2^\dagger & 0 \\
0 & U_2 U_1^\dagger
\end{pmatrix} 
(X\otimes I)(X\otimes I)(X\otimes I)
\begin{pmatrix}
U_2U_1^\dagger & 0 \\
0 & U_1 U_2^\dagger
\end{pmatrix}
= \\ 
\\
\begin{pmatrix}
(U_1U_2^\dagger)^2 & 0 \\
0 & (U_2 U_1^\dagger)^2
\end{pmatrix} 
(X\otimes I)
\in \mathcal{CH}_{k-2}.
\end{array}
\end{equation*}
By repeated application of this identity, we see that for $U=U_1\oplus U_2$ to be in $\mathcal{CH}$ all gates of the following form must also be in $\mathcal{CH}$ (or more specifically in $\mathcal{CH}_{k-m}$):
\begin{equation}\label{gateForm2}
\begin{array}{cc}
\begin{pmatrix}
(U_1U_2^\dagger)^{2^m} & 0 \\
0 & (U_2 U_1^\dagger)^{2^m}
\end{pmatrix} 
(X\otimes I)
\in \mathcal{CH}_{k-m}.
\end{array}
\end{equation}

Note that $(X\otimes I)$ can be ignored in Equ.~\ref{gateForm2} since multiplication by Clifford gates does not change the level of a gate in the Clifford Hierarchy.

Since we have assumed that $U$ is in $\mathcal{CH}$ at some finite level, for some choice of $m$ in Eqn.~\ref{gateForm2} the gate must be in the first level of $\mathcal{CH}$ which contains only Pauli gates. From Lemma~\ref{lem1} we conclude given such an $m$ that $(U_1U_2^\dagger)^{2^m} = \pm (U_2 U_1^\dagger)^{2^m}$ and that both $(U_1U_2^\dagger)^{2^m}$ and $\pm (U_2 U_1^\dagger)^{2^m}$ are Pauli.
\end{proof}
 
\begin{corollary}\label{cor2.4}
A controlled-unitary gate, ${\bf C}(U)$, is in the Clifford Hierarchy only if $U^{2^m}$ is Pauli for some $m>0$ and $U \in \mathcal{CH}$ at some finite level.
\end{corollary}

\begin{proof}
As previously discussed up to multiplication by Clifford gates, a controlled-unitary gate can be written as $I\oplus U$. We can then set $U_1 = I$ in Theorem \ref{thm2.1}. Our necessary conditions imply that for ${\bf C}(U)$ to be in the Clifford Hierarchy, we must have that $(U)^{2^m}$ is Pauli for some $m>0$. Up to a global phase all Pauli strings are order 2. Therefore, we have that (up to a phase) $(U)^{2^{m+1}} = I$ and we conclude that only controlled-unitary gates with $U \in \mathcal{CH}$ and with $U$ of order $2^m$ (for some integer $m>0$) can be in the Clifford Hierarchy.
\end{proof}

\section{But is it sufficient?}\label{sufficient}
In \cite{Hu2021ClimbingHierarchy} it was noted that one can `climb' the Clifford Hierarchy by starting with non-trivial Pauli matrices (rotations about the $X,Y,Z$ axis by $\pi$) and adding a control or taking a square root of the unitary rotation. In this section we denote these Pauli matrices by $X,Y,$ and $Z$.

Using the constraints derived above, we can look at what other single-qubit Clifford gates can climb the Clifford Hierarchy by adding a control. We will show that controlled single-qubit Clifford gates are in the Clifford Hierarchy if and only if they are order $2^k$ for some integer $k\ge 0$. 

We refer to the following set of single-qubit Clifford gates as Hadamard-like:
\begin{equation}
    \mathcal{S}_H\equiv \{\sigma_1 e^{i(\pi/4) \sigma_2}\;|\; \forall \sigma_1, \sigma_2 \in \{X,Y,Z\} \mbox{ and } \sigma_1 \ne \sigma_2\}. 
\end{equation}
Hadamard-like gates are all order two so a controlled Hadamard-like gate satisfies our necessary conditions for being in $\mathcal{CH}$. With a few circuit identities, it can be shown that adding a control to any element in $\mathcal{S}_H$ produces a gate in the third level of $\mathcal{CH}$. Furthermore, adding additional controls continues to ratchet any Hadamard-like gate up the Clifford Hierarchy by one level for each control. 

We refer to the following set of single-qubit Clifford gates as odd-order:
\begin{equation}
    \mathcal{S}_{O}\equiv\{\sigma_1 e^{i(\pi/4) \sigma_2} e^{i(\pi/4) \sigma_3}\;|\; \forall \sigma_1 \in \{I,X,Y,Z\}, \sigma_2,\sigma_3 \in \{X,Y,Z\}\mbox{ and } \sigma_2 \ne \sigma_3\}.
\end{equation}
It is easy to verify that all elements of $\mathcal{S}_O$ are odd-order. From this fact alone we can see that these gates cannot hope to climb the Clifford Hierarchy and can only dream of the summit above. Or less poetically: No controlled elements of $\mathcal{S}_O$ are contained in the Clifford Hierarchy. 

For completeness we define the set of order-four Clifford gates:
\begin{equation}
    \mathcal{S}_4\equiv \{e^{i(\pi/4) \sigma}\;|\; \forall \sigma \in \{X,Y,Z\}\}. 
\end{equation}
It is known that adding controls to elements in $\mathcal{S}_4$ creates gates in the third level of the Clifford Hierarchy\cite{Hu2021ClimbingHierarchy}.

The union of the distinct sets: $\mathcal{S}_H, \mathcal{S}_O,$ and $\mathcal{S}_4$, along with the single-qubit Pauli matrices, constitute all single-qubit gates in the Clifford group. Therefore, for one-qubit Clifford gates we have shown our conditions are necessary and sufficient. 

We can also check if our conditions are necessary for diagonal gates. The diagonal gates in the Clifford Hierarchy are known and they are are all of order $2^k$ for some integer $k\ge 0$. It is also known that adding a control to a diagonal gate in the Clifford Hierarchy produces a gate which is one level up in the Clifford Hierarchy \cite{Hu2021ClimbingHierarchy, Cui2017DiagonalHierarchy}. Hence, these conditions are also necessary and sufficient for diagonal gates. 

The complete structure of the Clifford Hierarchy is not fully known, nevertheless it would be interesting to prove the sufficiency of the necessary conditions found here. Even if sufficiency can be proven, we have not yet proven whether or not ${\bf C}(U)$, when $U$ satisfies the appropriate conditions, is always one level higher in the Clifford Hierarchy. Or more mathematically: If $U\in \mathcal{CH}_k\backslash \mathcal{CH}_{k-1}$, is ${\bf C}(U) \in \mathcal{CH}_{k+1}\backslash \mathcal{CH}_{k}$? We leave this as an open problem.   

\section{Extensions to Qudits}
Lemma \ref{lem1} and Theorem \ref{thm1.1} can easily be generalized to the (prime dimension, $d$) qudit case. Any controlled-$U$ gate which is block diagonal (with $d$, $d\times d$ blocks) with $U$ appearing explicitly in at least one $d\times d$ block, should suffice for Theorem \ref{thm1.1} to hold. At this time, we do not have a generalization of Theorem \ref{thm2.1} to qudits. It is tempting to conjecture for the qudit case that ${\bf C}(U)$ is in the Clifford Hierarchy only if $U$ is order $d^m$ where $d$ is the local qudit dimension and $m$ is a non-negative integer. We leave this as an open problem. 

\appendix

\section{Alternative Proofs}

Here we present a proof of Corollary~\ref{cor2.4} using circuit diagrams. This alternative proof may be more intuitive for some readers. 

First, we will make use of the following circuit identity which holds for any controlled-$U$ gate on qubits, ${\bf C}(U)$. We include it here as a lemma.

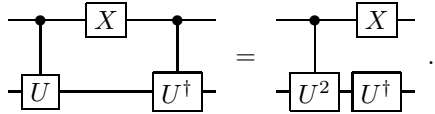
\begin{lemma}\label{circ}
\begin{equation*}
\begin{array}{ccc}
\Qcircuit @C=0.5em @R=1.3em {
& \ctrl{1} & \qw & \gate{X} & \qw & \ctrl{1} & \qw \\
& \gate{U} & \qw & \qw & \qw & \gate{U^\dagger} & \qw \\
}
\end{array}
= 
\begin{array}{ccc}
\Qcircuit @C=0.5em @R=1.3em {
& \ctrl{1} & \gate{X} & \qw \\
& \gate{U^2} & \gate{U^\dagger} & \qw \\
}
\end{array}.
\end{equation*}
\end{lemma}

\begin{theorem}\label{adx2}
${\bf C}(U)$, the controlled-$U$ gate, is in the Clifford Hierarchy only if $U$ is order $2^m$ for some integer $m>0$ and $U^{2^{m-1}}$ is a Pauli string (a member of $\mathcal{CH}_1$).\end{theorem}

An $n$-qubit gate, $V$, is in $\mathcal{CH}$ at level $k$ if $VPV^\dagger$ is in $\mathcal{CH}$ at level $k-1$ for all $n$-qubit Pauli strings $P$. 

For a gate, $U$ acting on $n$ qubits, to be in the Clifford Hierarchy at level $k$ ($\mathcal{CH}_k$), it must have $U P_i U^\dagger\in\mathcal{CH}_{k-1}$ for all $n$-qubit Pauli strings $P_i$. This requirement is both necessary and sufficient for membership in $\mathcal{CH}_k$.  

\begin{proof}
Now since a controlled gate has the same requirements for membership in $\mathcal{CH}_k$, we can see (using Lemma~\ref{circ} above) that for a controlled-$U$ gate to be in $\mathcal{CH}_k$ (for some finite level $k$) it is necessary that 

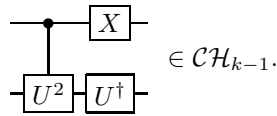
\begin{equation*}
\begin{array}{ccc}
\Qcircuit @C=0.5em @R=1.3em {
& \ctrl{1} & \gate{X} & \qw \\
& \gate{U^2} & \gate{U^\dagger} & \qw \\
}
\end{array}
\in \mathcal{CH}_{k-1}.
\end{equation*}

But for this to be true it must also hold that 
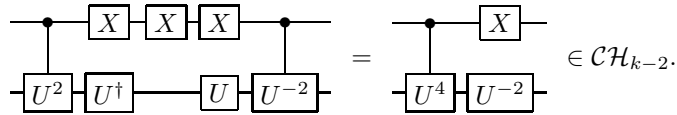
\begin{equation*}
\begin{array}{ccc}
\Qcircuit @C=0.5em @R=1.3em {
& \ctrl{1} & \gate{X} & \gate{X} & \gate{X} & \ctrl{1} & \qw \\
& \gate{U^2} & \gate{U^\dagger} & \qw  & \gate{U} & \gate{U^{-2}} & \qw\\
}
\end{array}
= 
\begin{array}{ccc}
\Qcircuit @C=0.5em @R=1.3em {
& \ctrl{1} & \gate{X} & \qw \\
& \gate{U^4} & \gate{U^{-2}} & \qw \\
}
\end{array}
\in \mathcal{CH}_{k-2}.
\end{equation*}
By repeated application of this identity, we see that for ${\bf C}(U)$ to be in $\mathcal{CH}$ all gates of the following form must also be in $\mathcal{CH}$:
\begin{equation}\label{gateForm}
\begin{array}{ccc}
\Qcircuit @C=0.5em @R=1.3em {
& \ctrl{1} & \gate{X} & \qw \\
& \gate{U^{2^m}} & \gate{U^{-{2^{m-1}}}} & \qw \\
}
\end{array}.
\end{equation}
Since we have assumed that this gate is in $\mathcal{CH}$ at some finite level, for some choice of $m$ in Eqn.~\ref{gateForm} the gate must be in the first level of $\mathcal{CH}$ which contains only Pauli gates. Since there are no (non-trivial) controlled gates at this level, we see that $U^{2^m}$ must equal identity (up to a global phase) and furthermore that $U^{-{2^{m-1}}}$ must be a Pauli string. Note that we can always choose an $m$ such that $U^{2^m}$ is the smallest $m$ where $U^{2^m}$ equals identity, then $U^{-{2^{m-1}}}$ must equal a non-identity Pauli string. We have now shown that for a gate ${\bf C}(U)$ to be in the Clifford Hierarchy, it is necessary that $U$ is order $2^m$ and that $U^{-{2^{m-1}}}$ is a Pauli string (a member of $\mathcal{CH}_1$). This concludes our proof. 
\end{proof}

\bibliographystyle{amsplain}


\end{document}